\documentclass[letterpaper,11pt]{article}
\usepackage{tabularx} 
\usepackage{amsmath}  
\usepackage{amsthm}
\usepackage{amsfonts}
\usepackage{tikz}
\usetikzlibrary{fit,graphs,graphs.standard,quotes,shapes,arrows,automata}
\usepackage{graphicx} 
\usepackage[margin=1.3in,letterpaper]{geometry} 
\usepackage{cite} 
\usepackage[final]{hyperref} 
\hypersetup{
	colorlinks=true,       
	linkcolor=red,        
	citecolor=blue,        
	filecolor=magenta,     
	urlcolor=blue         
}
\newtheorem{theorem}{Theorem}
\newtheorem{lemma}{Lemma}

\newtheorem{corollary}{Corollary}

\definecolor{bblue}{rgb}{0.2, 0.4, 0.8}
\definecolor{bgreen}{rgb}{0.2, 0.8, 0.2}
\definecolor{bred}{rgb}{0.8, 0.2, 0.2}
\definecolor{lgreen}{rgb}{0.0, 0.48, 0.0}
\definecolor{lpurple}{rgb}{0.48, 0.0, 0.48}

\tikzset{
	treenode/.style = {align=center, inner sep=0pt, text centered,
		font=\sffamily},
	arn_nn/.style = {treenode, circle, bblue, draw=bblue, 
		fill=bblue!10,
		minimum width=0.5em, minimum height=0.5em
	},
	arn_n/.style = {treenode, circle, bblue, draw=bblue, 
		text width=2.0em, very thick, 
		fill=bblue!10},
	arn_g/.style = {treenode, circle, bgreen, draw=bgreen, 
		text width=1.5em, very thick,
		fill=bblue!10},
	arn_r/.style = {treenode, circle, bred, draw=bred, 
		text width=2.0em, very thick,
		fill=bred!10},
	arn_x/.style = {treenode, rectangle, draw=black,
		minimum width=0.5em, minimum height=0.5em}
}

\begin{document}
	
\title{On Automata Recognizing Birecurrent Sets}
\author{Andrew Ryzhikov}
\date{LIGM, Universit\'e Paris-Est}

\maketitle

\begin{abstract} In this note we study automata recognizing birecurrent sets. A set of words is birecurrent if the minimal partial DFA recognizing this set and the minimal partial DFA recognizing the reversal of this set are both strongly connected. This notion was introduced by Perrin, and Dolce et al. provided a characterization of such sets. We prove that deciding whether a partial DFA recognizes a birecurrent set is a PSPACE-complete problem. We show that this problem is PSPACE-complete even in the case of binary partial DFAs with all states accepting and in the case of binary complete DFAs. We also consider a related problem of computing the rank of a partial DFA. 
\end{abstract}

\section{Introduction}

A partial deterministic finite automaton $A$ (which we further simply call a {\em partial automaton}) is a triple $A = (Q, \Sigma, \delta)$, where $Q$ is the set of its states, $\Sigma$ is a finite alphabet and $\delta: Q \times \Sigma \to Q$ is a partial transition function. Note that by default our definition of automata does not include initial and accepting states. If $\delta$ is defined for all pairs of $Q \times \Sigma$, $A$ is called a {\em complete automaton}. A state in an automaton is a {\em sink state} if all letters map this state to itself or is undefined for this state. A partial automaton is called {\em strongly connected} if for any ordered pair of states $s, t$ there exists a word mapping $s$ to $t$.

Given a partial automaton $A = (Q, \Sigma, \delta)$, the {\em rank} of a word $w \in \Sigma^*$ with respect to a set $S \subseteq Q$ (respectively, to the entire automaton $A$) is the size of the image of $S$ (respectively, $Q$) under the mapping defined by $w$ in $A$, i.e. the number $|\{\delta(s, w) \mid s \in S\}|$ ($|\{\delta(s, w) \mid s \in Q\}|$). The {\em rank} of a subset $S$ of states is the minimum nonzero integer among the ranks of all words $w \in \Sigma^*$ with respect to $S$. The rank of a partial automaton is the rank of the whole set of its states. We call a word of rank $1$ with respect to a partial automaton a {\em synchronizing word} for this automaton. A partial automaton having a synchronizing word is called {\em synchronizing} (this definition differs from the definition of a carefully synchronizing partial automaton \cite{Martyugin2014}, as well as from any of the $D_1$-, $D_2$-, or $D_3$-directing NFAs \cite{Martyugin2010}). This is a natural generalization of synchronizing complete automata to partial automata. Indeed, if the input sequence containing an error contains also a synchronizing word not affected by the error, after reading this word we can resynchronize, since we know the only possible active state for the correct sequence. The rank of a complete automaton can be found in polynomial time \cite{Rystsov1992}. Checking whether there exists a word of rank precisely $r$ with respect to a complete automaton is PSPACE-complete for any fixed $r \ge 2$ \cite{Goralcik1995}.

A set of words is {\em birecurrent} if the minimal partial automaton recognizing this set and the minimal partial automaton recognizing the reversal of this set are both strongly connected. This notion was introduced by Perrin in \cite{Perrin2013}. We study minimal partial automata (with initial and accepting states) recognizing birecurrent sets. This is equivalent to saying that the automaton and the determinization of the reversal of the automaton are both strongly connected \cite{Dolce2017}. A set $S \subseteq Q$ in a partial automaton $A = (Q, \Sigma, \delta)$ is said to be {\em saturated} by a word $w$ if $S$ is a pre-image of a subset of $\{\delta(q, w) \mid q \in Q\}$ under the mapping defined by $w$. In other words, $S$ is saturated by $w$, if under mapping defined by $w$ no state in $Q \setminus S$ is mapped to the set $\{\delta(s, w) \mid s \in S\}$, and each state in $S$ is mapped to some state. In \cite{Dolce2017} it is proved that a minimal strongly connected partial automaton $A$ (with initial and accepting states) recognizes a birecurrent set if and only if the set of its accepting states is saturated by a word $w$ such that the rank of $w$ with respect to $A$ is equal to the rank of $A$ (i.e., $w$ is a word of minimum rank with respect to $A$).

The examples of automata recognizing birecurrent sets include trim permutation automata (an automaton is called {\em permutation} if each its letter defines a bijection on the set of states) and automata recognizing a submonoid generated by a bifix code, see \cite{Dolce2017} for more examples.

In this note, we answer an open question stated in \cite{Dolce2017} by showing that checking whether a binary partial automaton recognizes a birecurrent set is PSPACE-complete, even in the cases where all states are accepting or the automaton is complete. We also investigate a related problem of computing the rank of a partial automaton.

\section{Rank of Partial Automata}

First, we prove a lemma that will be used to show that some problems about partial automata recognizing birecurrent sets are in PSPACE.

\begin{lemma} \label{lm-psp-rank}
	Computing the rank of a partial automaton can be done in polynomial space.
\end{lemma}
\begin{proof}
	First, note that checking whether there exists a word of rank $k$ can be done in polynomial space. Indeed, we can non-deterministically guess such a word (by applying it letter by letter) and check whether the size of the image of the set of states equals $k$. Then we use Savitch theorem stating that PSPACE = NPSPACE \cite{Sipser2006}. Now it remains to find such $k$ that there exists a word of rank $k$, and either $k = 1$ or there exists no word of rank less than $k$. This can be done by checking all possible values of $k$.
\end{proof}

In fact, computing the rank of a partial automaton is PSPACE-complete, and it is PSPACE-complete even to check that a partial automaton has rank $1$. To show that, we provide a reduction from the following PSPACE-complete problem \cite{Kozen1977}.

\begin{tabular}{||p{36em}}
	~{\sc Finite Automata Intersection} \\
	~{\em Input}: Complete automata $A_1, \ldots, A_k$ (with input and accepting states); \\
	~{\em Output}: Yes if there is a word which is accepted by all automata, No otherwise.
\end{tabular}

\begin{theorem} \label{thm-sync}
	Checking whether a partial automaton is synchronizing is PSPACE-complete.
\end{theorem}
\begin{proof}
	It follows from Lemma \ref{lm-psp-rank} that the problem is in PSPACE.
	
	We reduce the {\sc Finite Automata Intersection} to prove the hardness. Provided $A_1, \ldots, A_k$, $A_i = (Q_i, \Sigma, \delta_i)$, construct a partial automaton $A = (Q, \Sigma \cup \{r, z\}, \delta)$ as follows. The set $Q$ is the union of the sets of states of $A_1, \ldots, A_k$ together with two new states $Y, N$. We can assume that one of the automata in the input accepts all words in $\Sigma^*$, since adding such an automaton to the input does not change the answer. We define $\delta(q, x) = \delta_i(q, x)$ for $q \in Q_i$, $x \in \Sigma$. We also define $\delta(q, r) = s_i$ for all $q \in Q_i$, where $s_i$ is the initial state in $A_i$. We set $\delta(q, z) = Y$ if $q \in Q_i$ is an accepting state in $A_i$ and $\delta(q, z) = N$ if $q \in Q_i$ is not an accepting state in $A_i$. Finally, we define all letters in $\Sigma \cup \{r\}$ to act on $Y, N$ as self-loops. Thus, we leave undefined only the transitions for the letter $z$ from the states $Y, N$.
	
	We claim that $A$ is synchronizing if and ony if all $A_1, \ldots, A_k$ accept a common word. Indeed, if $w$ is a word accepted by all automata, then $rwz$ is a word synchronizing $A$. Conversely, assume that $A$ is synchronizing and $w$ is a word synchronizing $A$. If the letter $z$ does not occur in $w$, then the rank of $A$ is at least $k+2$. Thus, $w$ contains at least one occurence of $z$. After the first application of $z$, the set of reached states is a subset of $\{Y, N\}$ containing $Y$ (recall that we assume one of the input automata to accept all words in $\Sigma^*$). After applying another letter, the set of the reached states will become empty or unchanged, thus we can assume that $w = w_1z$ and $w_1$ does not contain $z$. The image of $Q$ under the mapping defined by $w_1$ is a subset of the set of accepting states of $A_1, \ldots, A_k$. Let $w_2$ be the suffix of $w_1$ after the last occurence of $r$ (or $w_1$ itself if $w_1$ does not contain $r$). Then the word $w_2$ is accepted by all the automata $A_1, \ldots, A_k$.   
\end{proof}

However, for strongly connected partial automata the situation is different. In fact, the rank of strongly connected partial automata behaves in the way similar to the rank of complete automata.

\begin{theorem}
	A word of minimum nonzero rank with respect to a strongly connected partial automaton can be found in polynomial time.
\end{theorem}
\begin{proof}
	Observe that in a strongly connected partial automaton any word of nonzero rank is a prefix of a word of minimum rank. Indeed, let $A = (Q, \Sigma, \delta)$ be a strongly connected partial automaton. Let $u$ be a word of nonzero rank and $v$ be a word of minimum nonero rank with respect to $A$. Let $q$ be an image of some state of $A$ under the mapping defined by $u$. Since the rank of $v$ is nonzero, there exists a state $q'$ in $A$ such that $\delta(q', v)$ is defined. Let $w$ be a word mapping $q$ to $q'$ (since $A$ is strongly connected, such word exists). Then $uwv$ is a word of minimum nonzero rank with respect to $A$.
	
	Now we describe a polynomial algorithm to find a word of minimum nonzero rank with resprect to $A$. We start with the whole set $S = Q$ of states and an empty word $w$, and at each step we try to to find a pair $q, q'$ of states and a word $v$ such that either exactly one of $\delta(q, v)$ and $\delta(q', v)$ is defined, or they both are defined and $\delta(q, v) = \delta(q', v)$. To find it, we solve a reachability problem in the subautomaton of the power automaton of $A$ induced by the subsets of size at most two \cite{Volkov2008}. If we find such pair of states, we add $v$ as a suffix to $w$ and proceed recursively with the set of states obtained by applying the mapping defined by $v$ to each state in $S$. 
	
	Observe that such pair of states exists if and only if the size of $S$ is more than the rank of $A$. Indeed, since each word of nonzero rank is a prefix of some word of minimum nonzero rank, we can always reduce the size of $S$ by applying a word of minimum nonzero rank if and only if $|S|$ is greater that the rank of $A$. Hence, the word $w$ thus constructed is a word of minimum nonzero rank with respect to $A$. 
\end{proof}

Observe that to merge two states in the described algorithm, we need a word of length at most $\frac{|Q|(|Q| - 1)}{2}$, since it is the number of different two-element subsets of the set of states. To map a state to another state we need a word of length at most $|Q| - 1$. Thus, we obtain the following.

\begin{corollary}
	The length of a shortest word of minimum nonzero rank $r$ in a $n$-state strongly connected partial automaton is at most $\frac{(n-1)((n-r)(n+2) - 2)}{2}$.
\end{corollary}

\section{Automata Recognizing Birecurrent Sets}

\begin{corollary}
	Checking whether a set of states of a partial automaton $A$ can be saturated by a word of minimum rank with respect to $A$ can be done in polynomial space.
\end{corollary}
\begin{proof}
	According to Lemma \ref{lm-psp-rank}, computing the rank of a partial automaton can be done in polynomial space. Thus, the problem of checking whether a set of states is saturated by a word of minimum rank is in NPSPACE, since we can non-deterministically guess a word saturating $S$ and having minimum nonzero rank with respect to $A$. All this conditions can be checked with polynomial space. Since NPSPACE = PSPACE \cite{Sipser2006}, the problem is in PSPACE.
\end{proof}

Now we are going to show that checking whether the whole set of states is saturated by a word of minimum rank is PSPACE-hard for strongly connected partial automata.

The following proof uses the idea of the proof from \cite{Martyugin2014} that careful synchronization of partial automata is PSPACE-complete.

\begin{theorem} \label{tm-partial}
	Checking whether the whole set of states of a partial automaton $A$ is saturated by a word of minimum nonzero rank with respect to $A$ is PSPACE-complete.
\end{theorem}
\begin{proof}
	We reduce the {\sc Finite Automata Intersection} problem. Given the automata $A_1, \ldots, A_k$, $A_i = (Q_i, \Sigma, \delta_i)$, we construct a partial automaton $A = (Q, \Sigma \cup \{r, z\}, \delta)$ as follows. We set $Q = \cup_{i = 1}^k Q_i \cup \{Y\}$, and define $\delta(q, x) = \delta_i(q, x)$ for $q \in Q_i$, $x \in \Sigma$. Next, we set $\delta(q, r) = s_i$ for $q \in Q_i$, where $s_i$ is the initial state of $A_i$, and $\delta(Y, r) = Y$. Finally, we define $\delta(q, z) = Y$ if $q \in Q_i$ is accepting in $A_i$, $\delta(q, z) = Y$ if $q = Y$, and leave $\delta(q, z)$ undefined otherwise.
	
	Let $w'$ be a word mapping the initial state of $A_1$ to some its accepting state (we can assume that such word exists without loss of generality). Then the word $rw'z$ has rank $1$ with respect to $A$. Thus, the automaton $A$ has rank $1$. 
	
	We shall prove that the set $Q$ is saturated by a word of rank $1$ if and only if all $A_1, \ldots, A_k$ accept a common word. Assume that the set $Q$ is saturated by a word $w$ of rank $1$. Then $\delta(q, w) = Y$ since $Y$ is a sink state, and moreover we can assume that the last letter of $w$ is $z$. Let $w = w_1z$. If $w_1$ does not contain $r$, then it is a word accepted by all automata $A_1, \ldots, A_k$. Otherwise, consider the last occurrence of $r$ in $w_1$ and let $w_2$ be the suffix of $w_1$ after this occurrence. Then $w_2$ is a word accepted by all automata $A_1, \ldots, A_k$.
	
	In the other direction, assume that the word $w$ is accepted by $A_1, \ldots, A_k$. Then the word $rwz$ is of rank $1$ and saturates the set $Q$. 
\end{proof}

\begin{theorem} \label{tm-partial-scb}
	Checking whether the whole set of states of a binary strongly connected partial automaton $A$ is saturated by a word of minimum nonzero rank with respect to $A$ is PSPACE-complete.
\end{theorem}
\begin{proof}
	We modify the reduction in Theorem \ref{tm-partial} to make the automaton strongly connected first. We can assume that in each automaton $A_1, \ldots, A_k$ all states are reachable from the initial state, and at least one accepting state is reachable from each state. Let $t_1, \ldots, t_m$ be the states in $A$ such that adding the transitions from $Y$ to each $t_i$ makes $A$ strongly connected. Define $m$ new letters $\ell_1, \ldots, \ell_m$ to map $Y$ to $t_1, \dots, t_m$ and undefined for all other states. Denote the automaton thus constructed as $A'$. The rank of $A'$ is $1$ by the same argument as in Theorem \ref{tm-partial}. Moreover, if $Q$ is saturated by a word of rank $1$, then each element of $Q$ is first mapped to $Y$, because no letter among $\ell_1, \ldots, \ell_m$ can be applied before that. Thus, the proof of Theorem \ref{tm-partial} works also for $A'$, and $Q$ is saturated by a word of minimum rank if and only if all $A_1, \ldots, A_k$ accept a common word.
	
	To make the automaton binary, we use the reduction described, for example, in Lemma 6 of \cite{Vorel2016}. Define the automaton $A_b = (Q_b, \{0, 1\}, \delta_b)$ as follows. Let $Q_b = Q \times \Sigma$. Let $\Sigma = \{x_1, \ldots, x_n\}$, where $x_n = r$. Define $\delta_b((q, x_i), 0) = (q, x_{i + 1})$ for $1 \le i \le n - 1$, and $\delta_b((q, x_n), 0) = (q, x_n)$. Set $\delta_b((q, x_i), 1) = (\delta(q, x_i), x_1)$. Informally, the letter $0$ emulates choosing a letter of $A'$, and the letter $1$ emulates applying the chosen letter. Moreover, applying the word $0^{n - 1}1$ maps the set $Q_b$ to the set $\{(q, x_1) \mid q \in Q\}$. Because of this correspondence, the automaton $A_b$ is strongly connected, has rank $1$, and satisfies all necessary conditions of the reduction.
\end{proof}

\begin{corollary}
	Checking whether a binary partial automaton where all states are accepting recognizes a birecurrent set is PSPACE-complete.
\end{corollary} 

By using the reduction described in the proof, the following corollary of Theorem \ref{thm-sync} can be proved. The only difference is that $x_n$ should be an additional letter acting as a self-loop on each state.

\begin{corollary}
	Checking whether a partial binary automaton is synchronizing is PSPACE-complete.
\end{corollary}

	

The next theorem uses the idea of the proof in \cite{Vorel2016} that checking whether a subset of states is synchronizing is PSPACE-complete in strongly connected complete automata.

\begin{theorem}
	Checking whether a subset of states of a binary strongly connected complete automaton $A$ is saturated by a word of minimum rank with respect to $A$ is PSPACE-complete.
\end{theorem}
\begin{proof}
	We modify the reduction from the proof of Theorem \ref{tm-partial}. Let $\overline{A} = (\overline{Q}, \Sigma, \overline{\delta})$ be a copy of $A = (Q, \Sigma, \delta)$. For each $q \in Q$ denote by $\overline{q}$ its copy in $\overline{Q}$. Define the automaton $B = (Q_B, \Sigma_B, \delta_B)$ as follows. Set $Q_B = Q \cup \overline{Q} \cup \{E, \overline{E}\}$.
	
	We can assume that in each automaton $A_1, \ldots, A_k$ all states are reachable from the initial state and at least one accepting state is reachable from each state, each automaton accepts at least one word, and none of the automata accepts the whole $\Sigma^*$. Let $t_1, \ldots, t_m$ be the states in $A$ such that adding the transitions from $Y$ to each $t_i$ makes $A$ strongly connected. Define $\Sigma_B$ to be $\Sigma$ together with $m + 1$ new letters $\ell_1, \ldots, \ell_m$ and $z$.
	
	For each $q \in Q_i$ which is not an accepting state of $A_i$, define $\delta_B(q, z) = E$, $\delta_B(\overline{q}, z) = \overline{E}$. Define $\delta_B(E, x) = E$, $\delta_B(\overline{E}, x) = \overline{E}$ for each $x \in \Sigma \cup \{z\}$. For each $q \in Q$, $x \in \Sigma$, define $\delta_B(q, x) = \delta(q, x)$, $\delta_B(\overline{q}, x) = \delta(\overline{q}, x)$. Next, for $1 \le i \le m$ define $\delta_B(Y, \ell_i) = t_i$ and $\delta_B(\overline{Y}, \ell_i) = \overline{t}_i$. Define $\delta_B(q, \ell_i) = \overline{t}_i$ and $\delta_B(\overline{q}, \ell_i) = t_i$ for each $1 \le i \le m$ and $q \in Q \setminus \{Y\} \cup \{E\}$. The automaton $B$ thus constructed is complete and strongly connected.
	
	Note that no pair of states $q, \overline{q}$ can be synchronized. Thus, $B$ has rank $2$, since the word $\ell_1$ has rank $2$ with respect to $B$. We shall prove that the set $S = Q \cup \{\overline{E}\}$ is saturated by a word of rank $2$ if and only if all automata $A_1, \ldots, A_k$ accept a common word.
	
	If the word $w$ is accepted by all automata $A_1, \ldots, A_k$, then the word $rw\ell_1$ saturates the set $S$. Assume now that the set $S$ is saturated by a word $w$ of rank $2$. Let $w = w_1xw_2$, where $w_1$ does not contain any of the letters $z, \ell_1, \ldots, \ell_m$, and $x$ is the first letter equal to $z$ or $\ell_i$. If the image of $Q$ under the mapping defined by $w_1$ has size at least two, then the state $E$ will be mapped to the image of some state of $S$, and thus $S$ is not saturated by $w$. Thus, $w_1$ maps $Q$ to one state, which means that the suffix of $w_1$ after the last occurrence of $r$ is accepted by all automata $A_1, \ldots, A_k$.
	
	To make the automaton binary, we again use the reduction described in the proof of Theorem \ref{tm-partial-scb}.
\end{proof}

\begin{corollary}
	Checking whether a binary complete automaton recognizes a birecurrent set is PSPACE-complete.
\end{corollary} 

\subsubsection*{Acknowledgments}

The author would like to thank Dominique Perrin for communicating the problem of finding the complexity of checking whether an automaton recognizes a birecurrent set to him, as well as for useful suggestions on the content and presentation of the paper.

{\footnotesize
\bibliography{SyncBib}
	\bibliographystyle{alpha}}
	
\end{document}